\documentclass[onecolumn,11pt]{article}
\usepackage[top=1in, bottom=1in, left=1.25in, right=1.25in]{geometry}
\usepackage[title]{appendix}

\usepackage{amssymb}
\usepackage{amssymb}
\usepackage{amsmath}
\usepackage{amssymb}
\usepackage[amsmath,thmmarks]{ntheorem}
\usepackage{indentfirst}

\usepackage[dvips]{graphicx}

\newtheorem{definition}{Definition}
\newtheorem{theorem}{Theorem}
\newtheorem{lemma}{Lemma}
\newtheorem{corollary}{Corollary}

\theoremheaderfont{\sc}\theorembodyfont{\upshape}
\theoremstyle{nonumberplain}
\theoremseparator{}
\theoremsymbol{\rule{1ex}{1ex}}
\newtheorem{proof}{Proof}

\begin{document}

\title{Proof of Convergence and Performance Analysis for Sparse Recovery via Zero-point Attracting Projection}

\author{Xiaohan~Wang,~Yuantao~Gu\thanks{This work was partially supported by National Natural Science Foundation of China (NSFC 60872087 and NSFC U0835003). The authors are with the Department of Electronic Engineering, Tsinghua University, Beijing 100084, China. The corresponding author of this paper is Yuantao Gu (Email: gyt@tsinghua.edu.cn).},~and~Laming~Chen}

\date{Received July 15, 2011; accepted Mar. 31, 2012.\\\vspace{1em}
This article appears in \textsl{IEEE Transactions on Signal Processing}, 60(8): 4081-4093, 2012.}

\maketitle

\begin{abstract}
A recursive algorithm named Zero-point Attracting Projection (ZAP) is proposed recently for sparse
signal reconstruction. Compared with the reference algorithms, ZAP demonstrates rather good
performance in recovery precision and robustness. However, any theoretical analysis about the
mentioned algorithm, even a proof on its convergence, is not available. In this work, a strict
proof on the convergence of ZAP is provided and the condition of convergence is put forward. Based
on the theoretical analysis, it is further proved that ZAP is non-biased and can approach the
sparse solution to any extent, with the proper choice of step-size. Furthermore, the case of
inaccurate measurements in noisy scenario is also discussed. It is proved that disturbance power
linearly reduces the recovery precision, which is predictable but not preventable. The
reconstruction deviation of $p$-compressible signal is also provided.
Finally, numerical simulations are performed to verify the theoretical analysis.

\end{abstract}

\section{Introduction}

\subsection{Overview of CS and Sparse Signal Recovery}

Compressive Sensing (CS)~\cite{compressivesampling,compressedsensing} is proposed as a novel
technique in the field of signal processing. Based on the sparsity of signals in some typical
domains, this method takes global measurements instead of samples in signal acquisition. The theory
of CS confirms that the measurements required for recovery are far fewer than conventional signal
acquisition technique.

With the advantages of sampling below Nyquist rate and little loss in reconstruction quality, CS
can be widely applied in the regions such as source coding~\cite{coding}, medical
imaging~\cite{MRI}, pattern recognition~\cite{pattern}, and wireless communication~\cite{wireless}.

Suppose that an $N$-dimensional vector ${\bf x}\in\mathbb{R}^N$ is a sparse signal with sparsity $S$,
which means that only $S$ entries of ${\bf x}$ are nonzero among all $N$ elements. An $M\times N$
measurement matrix ${\bf A}$ with $M<N$ is applied to take global measurements of ${\bf
x}$. Consequently an $M\times 1$ vector
\begin{equation}\label{yAx}
{\bf y}={\bf Ax}
\end{equation}
is obtained and the information of $N$-dimensional unknown signal is reduced to the
$M$-dimensional measurement vector. Exploiting the sparse property of ${\bf x}$, the original signal
can be reconstructed through ${\bf y}$ and ${\bf A}$.

The procedure of CS mainly includes two stages: signal measurement and signal reconstruction. The
key issues are the design of measurement matrix and the algorithm of sparse signal reconstruction,
respectively.

On the signal reconstruction of CS, a key problem is to derive the sparse solution, i.e., the
solution to the under-determined linear equation which has the minimal $\ell_0$ norm,
\begin{equation}\label{l0}
\min_{\bf x}\|{\bf x}\|_0, \quad\textrm{subject to } {\bf y=Ax}. \tag{$P_0$}
\end{equation}

However, (\ref{l0}) is a Non-deterministic Polynomial (NP) hard problem. It is
demonstrated that under certain conditions~\cite{compressedsensing}, (\ref{l0}) has the
same solution as the relaxed problem
\begin{equation}\label{l1}
\min_{\bf x}\|{\bf x}\|_1, \quad\textrm{subject to }{\bf y=Ax}.
\tag{$P_1$}
\end{equation}
(\ref{l1}) is a convex problem and can be solved through convex optimization.

In non-ideal scenarios, the measurement vector ${\bf y}$ is inaccurate with noise perturbation and
(\ref{yAx}) never satisfies exactly. Consequently, (\ref{l1}) is modified to
\begin{equation}\label{l1_eps}
\min_{\bf x}\|{\bf x}\|_1, \quad\textrm{subject to }\|{\bf y-Ax}\|_2\le
\varepsilon, \tag{$P_2$}
\end{equation}
where $\varepsilon$ is a positive number representing the energy of noise.

Many algorithms have been proposed to recover the sparse signal from $\bf y$ and $\bf A$. These
algorithms can be classified into several main categories, including greedy pursuit,
optimization algorithms, iterative thresholding algorithms and other
algorithms.

The greedy pursuit algorithms always choose the locally optimal approximation to the sparse
solution iteratively in each step. The computation complexity is low but more
measurements are needed for reconstruction. Typical algorithms include Matching Pursuit
(MP)~\cite{orthogonalmatching}, Orthogonal Matching Pursuit (OMP)~\cite{signalrecovery,omp},
Stage-wise OMP (StOMP)~\cite{sparsesolution}, Regularized OMP
(ROMP)~\cite{romp,signalrecoveryfrom}, Compressive Sampling MP (CoSaMP)~\cite{cosamp}, Subspace
Pursuit (SP)~\cite{subspacepursuit}, and Iterative Hard Thresholding (IHT)~\cite{iht}.

Optimization algorithms solve convex or non-convex problems and can be further divided into convex optimization and non-convex
optimization.
Convex optimization methods have the properties of fewer
measurements demanded, higher computation complexity, and more theoretical support in mathematics.
Convex optimization algorithms include Primal-Dual interior method for Convex Objectives (PDCO)~\cite{pdco}, Least Square QR (LSQR)~\cite{lsqr}, Large-scale $\ell_1$-regularized Least Squares
($\ell_1$-$ls$)~\cite{amethod}, Least Angle Regression (LARS)~\cite{LARS}, Gradient Projection for Sparse Reconstruction (GPSR)~\cite{gradientprojection}, Sparse Reconstruction by Separable Approximation
(SpaRSA)~\cite{sparsereconstruction}, Spectral Projected-Gradient $\ell_1$ (SPGL1)~\cite{spgl1}, Nesterov Algorithm (NESTA)~\cite{nesta} and Constrained Split Augmented Lagrangian Shrinkage Algorithm (C-SALSA)~\cite{csalsa}.

Non-convex optimization methods solve the problem of optimization by minimizing $\ell_p$ norm with
$0\le p<1$, which is not convex. This category of algorithms demands fewer measurements than convex
optimization methods. However, the non-convex property may lead to converging towards the local
extremum which is not the desired solution. Moreover, these methods have higher computation
complexity. Typical non-convex optimization methods are focal underdetermined system solver (FOCUSS)~\cite{focuss}, Iteratively Reweighted Least Square (IRLS)~\cite{lrls} and $\ell_0$ Analysis-based Sparsity (L0AbS)~\cite{l0abs}.

A new kind of method, Zero-point Attracting Projection (ZAP), has been recently proposed to solve
(\ref{l0}) or (\ref{l1})~\cite{astochastic}. The projection of the zero-point attracting term is
utilized to update the iterative solution in the solution space. Compared with the other
algorithms, ZAP has advantages of faster convergence rate, fewer measurements demanded, and a
better performance against noise.

However, ZAP is proposed with heuristic and experimental methodology and lacks a strict proof of
convergence~\cite{astochastic}. Though abundant computer simulations verify its performance, it is
still essential to prove its convergence, provide the specific working condition, and analyze
performances theoretically including the reconstruction precision, the convergence rate and the
noise resistance.

\subsection{Our Work}

This paper aims to provide a comprehensive analysis for ZAP. Specifically, it studies $\ell_1$-ZAP,
which uses the gradient of $\ell_1$ norm as the zero-point attracting term.
$\ell_0$-ZAP is non-convex and its convergence will be addressed in future work.

The main contribution of this work is to prove the convergence of $\ell_1$-ZAP in non-noisy
scenario. Our idea is summarized as follows. Firstly, the distance between the iterative solution
of $\ell_1$-ZAP and the original sparse signal is defined to evaluate the convergence. Then we
prove that such distance will decrease in each iteration, as long as it is larger than a constant
proportional to the step-size. Therefore, it is proved that $\ell_1$-ZAP is convergent to the
original sparse signal under non-noisy case, which provides a theoretical foundation for the
algorithm. Lemma 1 is the crucial contribution of this work, which
reveals the relationship between $\ell_1$ norm and $\ell_2$ norm in the solution space.

Another contribution is about the signal reconstruction with measurement noise. It is demonstrated
that $\ell_1$-ZAP can approach the original sparse signal to some extent under inaccurate
measurements. In the noisy case, the recovery precision is linear with not only the
step-size but also the energy of noise.

Other contributions include the discussions on some related topics. The convergence
rate is estimated as an upper bound of iteration number. The constraint of
initial value and its influence on convergence are provided. The convergence of
$\ell_1$-ZAP for $p$-compressible signal is also discussed. Experiment results are provided to
verify the analysis.

At the time of revising this paper, we are noticed of a similar algorithm called projected subgradient method~\cite{subgradient}, which leads to
some related researches~\cite{infeasible}. Though obtained from different frameworks,
$\ell_1$-ZAP shares the same recursion with the other. However, the two algorithms are not
exactly the same. The attracting term of ZAP is not restricted to the subgradient of a objective function, and can be used to solve either a convex
problem or a non-convex one, while only the subgradient of a convex function is
allowed in the mentioned method. Furthermore, the available analysis of the projected
subgradient method studies the convergence of the objective function, while this work focuses on the
properties of the iterative sequence, as derived from the significant Lemma 1. The theoretical analysis in this work may contribute to promoting the projected subgradient method.

The remainder of this paper is organized as follows. In Section II, some preliminary knowledge is
introduced to prepare for the main theorems. The main contribution in non-noisy scenario is
presented as Theorem 4 in Section III, which proves the convergence of $\ell_1$-ZAP. Some related
topics about Theorem 4 are also discussed in Section III. Section IV shows another main theorem in
noisy scenario, and some discussions are also brought out. Experiment results are shown in Section
V. The whole paper is concluded in Section VI.

\section{Preliminaries}

\subsection{RIP and Coherence}

In this subsection, Restricted Isometry Property (RIP) and coherence are introduced and
then some theorems on (\ref{l1}) and (\ref{l1_eps}) are presented, which will be helpful to the
following content.

\begin{definition}\cite{decodingby} Suppose ${\bf A}_{\mathcal T}$ is the $M\times |\mathcal T|$ submatrix by extracting the columns of
$M\times N$ matrix ${\bf A}$ corresponding to the indices in set $\mathcal T\subset
\{1,2,\dots,N\}$. The RIP constant $\delta_S$ is defined as the smallest nonnegative
quantity such that
$$
(1-\delta_S)\|{\bf c}\|_2^2\le \|{\bf A}_{\mathcal T}{\bf c}\|_2^2\le (1+\delta_S)\|{\bf c}\|_2^2
$$
holds for all subsets $\mathcal T$ with $|\mathcal T|\le S$ and vectors ${\bf
c}\in\mathbb{R}^{|\mathcal T|}$.
\end{definition}

\begin{theorem}\cite{therestricted}
If the RIP constant of matrix ${\bf A}$ satisfies the condition
\begin{equation}\label{eq23}
\delta_{2S}<\sqrt{2}-1,
\end{equation}
where $S$ is the sparsity of ${\bf x}$, then the solution of (\ref{l1}) is unique and
identical to the original signal.
\end{theorem}

\begin{theorem}\cite{therestricted}
If the RIP constant of matrix ${\bf A}$ satisfies the condition
\begin{equation}\label{eq47}
\delta_{2S}<\sqrt{2}-1,
\end{equation}
then the solution ${\bf x}^{\star}$ of (\ref{l1_eps}) obeys
\begin{equation}\label{eq48}
\|{\bf x}^{\star}-{\bf x}^{\sharp}\|_2\le C_S \varepsilon,
\end{equation}
where ${\bf x}^{\sharp}$ is the original signal of sparsity $S$ and $C_S$ is a positive
constant related to $S$.
\end{theorem}

RIP determines the property of the measurement matrix. Recent results on RIP can be found
in~\cite{improvedbounds,howsharp}.

\begin{definition}\cite{sparsityand}
The coherence of an $M\times N$ matrix ${\bf A}$ is defined as
$$
\mu({\bf A})=\max_{i\neq j}{|{\boldsymbol \alpha}_i^{\rm T}{\boldsymbol \alpha}_j|},
$$
where ${\boldsymbol \alpha}_i(1\le i \le N)$ is the $i$th column of
${\bf A}$ and $\|{\boldsymbol \alpha}_i\|_2=1$.
\end{definition}

\begin{theorem}\cite{gradientprojection,coherencebased}
If the sparsity $S$ of ${\bf x}$ and the coherence of matrix ${\bf A}$ satisfy the condition
\begin{equation}\label{eq49}
S<\frac{1}{3\mu({\bf A})},
\end{equation}
then the solution of (\ref{l1_eps}) is unique.
\end{theorem}

Theorem 1 provides the sufficient condition on exact recovery of the original signal without any
perturbation. It is also a loose sufficient condition of the unique solution of (\ref{l1}). Theorem
2 indicates that under the condition (\ref{eq47}), the solution of (\ref{l1_eps}) is not too far
from the original signal, with a deviation proportional to the energy of measurement noise.
Theorem 3 provides a sufficient condition of the uniqueness of the solution of (\ref{l1_eps}).

\subsection{$\ell_1$-ZAP}

In ZAP algorithm, the zero-point attracting term is used to update the iterative solution and then
the updated iterative solution is projected to the solution space. The procedures of ZAP can be
summarized as follows.

\vspace{0.5em}

Input: ${\bf A}\in \mathbb{R}^{M\times N},\ {\bf y}\in\mathbb{R}^M,\ \gamma\in\mathbb{R}_{+}$.

Initialization: $n = 0$ and ${\bf x}_0 = {\bf A}^{\dagger}{\bf y}$.

Iteration:

\quad while stop condition is not satisfied

\quad 1. Zero-point attraction:
\begin{equation}\label{eq11}
\hat{\bf x}_{n+1}={\bf x}_n-\gamma\cdot \nabla {{\rm F}({\bf x}_n)}
\end{equation}

\quad 2. Projection:
\begin{equation}\label{eq22}
{\bf x}_{n+1}=\hat{\bf x}_{n+1}+{\bf A}^{\dagger}({\bf y}-{\bf A}\hat{\bf x}_{n+1})
\end{equation}

\quad 3. Update the index: $n=n+1$

\quad end while

\vspace{0.5em}

In the initialization and (\ref{eq22}), ${\bf A}^{\dagger}={\bf A}^{\rm T}({\bf AA}^{\rm T})^{-1}$
denotes the pseudo-inverse of ${\bf A}$. In (\ref{eq11}), $\nabla {{\rm F}({\bf x}_n)}$ is the
zero-point attracting term, where ${\rm F}({\bf x})$ is a function representing the sparse penalty
of vector ${\bf x}$. Positive parameter $\gamma$ denotes the step-size in the step of zero-point
attraction.

ZAP was firstly proposed in \cite{astochastic} with a specification of $\ell_0$-norm constraint,
termed $\ell_0$-ZAP, in which the approximate $\ell_0$ norm is utilized as the function
F(x). $\ell_0$-ZAP belongs to the non-convex optimization methods and has an outstanding
performance beyond conventional algorithms. In \cite{astochastic}, the penalty function is $\|{\bf x}\|_0$ and its gradient is approximated as
$$
\nabla {{\rm F}_{\ell_0}({\bf x})}\approx [{\rm f}(x_1),{\rm f}(x_2),\cdots,{\rm f}(x_N)]^{\rm T}
$$
and
$$
{\rm f}(x)=\left\{
    \begin{array}{cl} -\alpha^2x-\alpha, & -\frac{1}{\alpha}\le x< 0; \\
    -\alpha^2x+\alpha, & 0<x\le\frac{1}{\alpha}; \\
    0, & {\rm elsewhere}. \end{array} \right.
$$
The piecewise and non-convex zero-point attracting term further increases the difficulty to theoretically analyze the convergence of $\ell_0$-ZAP.

As another variation of ZAP, $\ell_1$-ZAP is analyzed in this work. The function ${\rm F}({\bf x})$
is the $\ell_1$ norm of ${\bf x}$ in the zero-point attracting term. Since it is
non-differentiable, the gradient of ${\rm F}({\bf x})$ can be replaced by its sub-gradient.
Considering that the gradient of ${\rm F}({\bf x})$ is ${\rm sgn}({\bf x})$ when none of the
components of ${\bf x}$ are zero, (\ref{eq11}) can be specified as
\begin{equation}\label{eq62}
\hat{\bf x}_{n+1}={\bf x}_n-\gamma\cdot {\rm sgn}({\bf x}_n),
\end{equation}
where the gradient is replaced by one of the sub-gradients ${\rm sgn}({\bf x})$. The sign function
${\rm sgn}({\bf x})$ has the same size with ${\bf x}$ and each entry of ${\rm sgn}({\bf x})$ is the
scalar sign function of the corresponding entry of ${\bf x}$.

Experiments show that though its performance is better than conventional algorithms, $\ell_1$-ZAP
behaves not as good as $\ell_0$ norm constraint variation. However, as a convex optimization
method, $\ell_1$-ZAP has advantages beyond non-convex methods, as mentioned in introduction.
$\ell_1$-ZAP is considered in this paper as the first attempt to analyze ZAP in theory.

The steps (\ref{eq62}) and (\ref{eq22}) of $\ell_1$-ZAP can be combined into the following
recursion
\begin{equation}\label{eq5}
{\bf x}_{n+1}={\bf x}_n-\gamma {\bf P}{\rm sgn}({\bf x}_n)
\end{equation}
with the projection matrix
\begin{equation}\label{eq4}
{\bf P=I}-{\bf A}^{\rm T}({\bf AA}^{\rm T})^{-1}{\bf A}.
\end{equation}

Notice that following (\ref{eq5}), (\ref{eq4}) and the initialization, the sequence has the
property
\begin{equation}\label{eq57}
{\bf A}{\bf x}_{n+1}={\bf A}{\bf x}_n = {\bf A}{\bf x}_0 ={\bf y}, \quad\forall n\ge0,
\end{equation}
which means all iterative solutions fall in the solution space.

Numerical simulations demonstrate that the sparse solution of under-determined linear
equation can be calculated by $\ell_1$-ZAP. In fact, the sequence $\{{\bf x}_n\}$
calculated through (\ref{eq5}) is not strictly  convergent. $\{{\bf x}_n\}$ will fall
into the neighborhood of ${\bf x}^*$ after finite iterations, with radius proportional to step-size
$\gamma$. With the
increasing of iterations, ${\bf x}_n$ approaches ${\bf x}^*$ step by step at first.
However, it vibrates in the neighborhood of ${\bf x}^*$ when ${\bf x}_n$ is close enough
to ${\bf x}^*$. If the step-size $\gamma$ decreases, the radius of neighborhood also decreases.
Consequently, one can get the approximation to the sparse solution at any precision by choosing
appropriate step-size.

In this work the convergence of $\ell_1$-ZAP is proved. The main results are the following theorems
in Section III and IV, corresponding to non-noisy scenario and noisy scenario, respectively.

\section{Convergence in Non-Noisy Scenario}

The main contribution is included in this section. A lemma is proposed in Subsection A for
preparing the main theorem in Subsection B. Then the condition of exact signal recovery by
$\ell_1$-ZAP is given in Subsection C. Several constants and variables in the proof of
convergence are discussed in Subsections D and E. In Subsection F, an estimation on
the convergence rate is given. The initial value of $\ell_1$-ZAP is discussed in Subsection
G.

\subsection{Lemma}
\begin{lemma}\label{Lemma2}
Suppose that ${\bf x}\in\mathbb{R}^N$ satisfies ${\bf y}={\bf Ax}$, with given ${\bf
A}\in\mathbb{R}^{M\times N}$ and ${\bf y}\in\mathbb{R}^M$. ${\bf x}^*$ is the unique solution of
(\ref{l1}). If $\|{\bf x}-{\bf x}^*\|_2$ is bounded by a positive constant $M_0$, then there exists
a uniform positive constant $t$ depending on ${\bf A, y,}$ and $M_0$, such that
\begin{equation}\label{eqinlemma2}
\|{\bf x}\|_1-\|{\bf x}^*\|_1\ge t\|{\bf x}-{\bf x}^*\|_2
\end{equation}
holds for arbitrary ${\bf x}$ satisfying ${\bf y}={\bf Ax}$.
\end{lemma}

The outline of the proof is presented here while the details are included in Appendix A.

\begin{proof}
By defining
\begin{equation}\label{eq1}
{\rm g}({\bf x}) = \frac{\|{\bf x}\|_1-\|{\bf x}^*\|_1}{\|{\bf x}-{\bf x}^*\|_2},
\end{equation}
equation (\ref{eqinlemma2}) is equivalent to the following inequality
\begin{align}\label{eq6}
&\inf_{\bf x}{\rm g}({\bf x})>0,\quad\text{subject to}\;{\bf y=Ax}\;\text{and}\;0<\|{\bf x}-{\bf x}^*\|_2\le M_0.
\end{align}
Define the index set $\mathcal I=\{k ~|~ x^*_k\neq 0, 1\le k\le N\}$, then there exists a
positive constant $r_0$ such that $\left({\rm sgn}({\bf x})\right)_{\mathcal I}=\left({\rm
sgn}({\bf x}^*)\right)_{\mathcal I}$, when ${\bf x}$ satisfies
\begin{equation}\label{eq32}
\|{\bf x}-{\bf x}^*\|_2<r_0.
\end{equation}
The above proposition means that ${\bf x}$ and ${\bf x}^*$ share the same sign for the entries
indexed by $\mathcal I$. Define sets ${\mathcal X}_1$ and ${\mathcal X}_2$ as
\begin{align}\label{eq15}
{\mathcal X}_1&=\{{\bf x}~|~r_0\le\|{\bf x}-{\bf x}^*\|_2\le M_0\}\cap\{{\bf x}~|~{\bf y=Ax}\},\nonumber\\
{\mathcal X}_2&=\{{\bf x}~|~0<\|{\bf x}-{\bf x}^*\|_2<r_0\}\cap\{{\bf x}~|~{\bf y=Ax}\}.
\end{align}
Consequently, for the separate cases of ${\bf x}\in {\mathcal X}_1$ and ${\bf x}\in {\mathcal
X}_2$, it is proved that ${\rm g}({\bf x})$ has a positive lower bound, respectively. Combining
the two cases,  Lemma \ref{Lemma2} is proved.
\end{proof}

\subsection{Main Result}
\begin{theorem}\label{Theorem4}
Suppose that ${\bf x}^*$ is the unique solution of (\ref{l1}). ${\bf x}_{n+1}$ and ${\bf x}_n$
satisfy the recursion (\ref{eq5}) and ${\bf x}_n$ is energy constrained by $\|{\bf x}_n-{\bf
x}^*\|_2\le M_0$, where $M_0$ is a positive constant. Then the iteration obeys
\begin{equation}\label{eq9}
\|{\bf x}_{n+1}-{\bf x}^*\|_2^2\le \|{\bf x}_n-{\bf x}^*\|_2^2-d \gamma^2
\end{equation}
when
\begin{equation}\label{eq99}
\|{\bf x}_n-{\bf x}^*\|_2\ge K \gamma,
\end{equation}
where
\begin{align}\label{eq25}
K&=\frac{\mu}{2t}\max_{{\bf x}\in\mathbb{R}^N}{\|{\bf P}{\rm sgn}({\bf x})\|_2^2},\\
\label{eq26}
d&=(\mu-1)\max_{{\bf x}\in\mathbb{R}^N}\|{\bf P}{\rm sgn}({\bf x})\|_2^2
\end{align}
are two constants with a parameter $\mu > 1$, and $t>0$ denotes the lower bound specified in
Lemma~1.
\end{theorem}

For a given under-determined constraint (\ref{yAx}) and the unique sparsest solution of (\ref{l1}),
Theorem \ref{Theorem4} demonstrates the convergence property and provides the convergence
conditions of $\ell_1$-ZAP. As long as the iterative result ${\bf x}_n$ is far away from the sparse
solution ${\bf x}^*$, the new result ${\bf x}_{n+1}$ in next iteration affirmatively becomes closer
than its predecessor. Furthermore, the decrease in $\ell_2$ distance is a constant $d\gamma^2$,
which means ${\bf x}_n$ will definitely get into the $(K\gamma)$-neighborhood of ${\bf x}^*$ in
finite iterations. According to the definition of $K$, ${\bf x}_n$ can
approach the sparse solution ${\bf x}^*$ to any extent if the step-size $\gamma$ is chosen small
enough. Therefore, $\ell_1$-ZAP is convergent, i.e., the iterative result can get close to the
sparse solution at any precision. Here $\mu$ is a tradeoff parameter which balances the estimated
precision and convergence rate.

The proof of Theorem \ref{Theorem4} goes in Appendix B.

\subsection{Exact Signal Recovery by $\ell_1$-ZAP}

Using Theorem 4 and conditions added, the convergence of $\ell_1$-ZAP can be deduced, as the
following corollary.

\begin{corollary}
Under the condition (\ref{eq23}), $\ell_1$-ZAP can recover the original signal at any precision if
the step-size $\gamma$ can be chosen small enough.
\end{corollary}

\begin{proof}
Firstly, it will be demonstrated that the condition of energy constraint in Theorem 4 can always be
satisfied. In fact, $M_0$ can be chosen greater than $\|{\bf x}_0-{\bf x}^*\|_2$. If the energy
constraint $\|{\bf x}_n-{\bf x}^*\|_2<M_0$ holds for index $n$, the conditions of Theorem 4 are
satisfied and then $\|{\bf x}_{n+1}-{\bf x}^*\|_2<M_0$ holds naturally according to (\ref{eq9}).
Consequently, it is readily accepted that the condition of energy constraint is satisfied for each
index $n$, with the utilization of Theorem 4 in each step.

Combining the explanation after Theorem 4, it is clear that the $\ell_1$-ZAP is convergent to the
solution of (\ref{l1}) at any precision as long as the step-size is chosen small enough.

According to Theorem 1, it is known that under the condition of (\ref{eq23}), the
solution of (\ref{l1}) is unique and identical to the original sparse signal. Then
Corollary 1 is proved.
\end{proof}

According to Theorem 4 and Corollary 1, the sequence will surely get into the
$(K\gamma)$-neighborhood of ${\bf x}^*$. In fact, because of several inequalities used in the
proof, $K\gamma$ is merely a theoretical radius with conservative estimation. The actual
convergence may get into a even smaller neighborhood. The details will be discussed in
Subsection F.

\subsection{Constant $t, K,$ and the Extremum of $\|{\bf P}{\rm sgn}({\bf x})\|_2$}

Involved in (\ref{eq25}) of Theorem 4, constant $t$ is essential to the convergence of
$\ell_1$-ZAP. In fact, the key contribution of this work is to indicate the existence of this
constant. However, one can merely obtain the existence of $t$ from the proof of Lemma 1, other than
its exact value. Because ${\bf x}^*$ in the definition of (\ref{eq6}) is unknown, it is difficult
to give the exact value or formula of $t$, even though it is actually determined by ${\bf A}$,
${\bf y}$, and $M_0$. Whereas, an upper bound is given with some information about $t$,
which leads to Theorem 5.

According to (\ref{eq25}), constant $K$ is inversely proportional to $t$. With a small $t$, the
radius of convergent neighborhood is large and the convergence precision is worse. The maximum of
$\|{\bf P}{\rm sgn}({\bf x})\|_2$ is also involved in the definition of $K$. According
to the range of sign function, i.e. $\{-1,0,1\}$, there are $3^N$ choices of vector ${\rm sgn}({\bf
x})$ altogether. Similar to $t$, the extremum of $\|{\bf P}{\rm sgn}({\bf x})\|_2$ is determined by
$\bf A$.

The relationship between $t$ and extremum of $\|{\bf P}{\rm sgn}({\bf x})\|_2$ is presented in
Theorem 5.

\begin{theorem}
If $t$ is defined by (\ref{eq6}), one has the following inequality
\begin{equation}\label{eq24}
t\le \min_{{\bf x}\in {\mathcal X}_1\cup{\mathcal X}_2}{\|{\bf P}{\rm sgn}({\bf
x})\|_2}\le \max_{{\bf x}\in\mathbb{R}^N}{\|{\bf P}{\rm sgn}({\bf x})\|_2}\le \sqrt{N}.
\end{equation}
\end{theorem}

The proof of Theorem 5 is postponed to Appendix C.

According to the theorem, the minimum of $\|{\bf P}{\rm sgn}({\bf x})\|_2$ restricts the value of
$t$, as leads to worse precision of $\ell_1$-ZAP. Hence, the measurement matrix ${\bf A}$ should be
chosen with relatively large $\min\|{\bf P}{\rm sgn}({\bf x})\|_2$ to improve the
performance of the mentioned algorithm. The mathematical meaning of ${\bf P}{\rm
sgn}({\bf x})$ is the projection of ${\rm sgn}({\bf x})$ to the solution space of ${\bf y=Ax}$. For
a particular instance, if there exists a sign vector, to whom the solution space is almost
orthogonal, then the minimum of $\|{\bf P}{\rm sgn}({\bf x})\|_2$ is rather small and the precision
of convergence is bad. An additional explanation is that the solution space can not be strictly
orthogonal to any sign vector, or else it will lead to a contradiction with the condition of
(\ref{eq23}), i.e., the uniqueness of ${\bf x}^*$.

\subsection{Discussions on $\mu$ and Bound Sequence}

A parameter $\mu$ is involved in Theorem 4. We will discuss the choice of $\mu$ and some related
problems. First of all, it needs to be stressed that $\mu$ is just a parameter for the
bound sequence in theoretical analysis, other than a parameter for actual iterations.

According to the proof in Appendix B, as long as $\mu$ is chosen satisfying the conditions of
$\mu>1$ and
$$
\|{\bf x}_n-{\bf x}^*\|_2\ge \gamma\frac{\mu}{2t}\max_{{\bf x}\in\mathbb{R}^N}{\|{\bf P}{\rm sgn}({\bf x})\|_2^2},
$$
Theorem 4 holds and the distance between ${\bf x}_n$ and ${\bf x}^*$ decreases in the
next iteration. However, considering the expression of (\ref{eq26}), the decrease of $\|{\bf
x}_n-{\bf x}^*\|_2^2$ by each iteration is different for various $\mu$. There are two
strategies to choose the parameter $\mu$, a constant or a variable one.

When $\mu$ is chosen as a constant, Theorem 4 indicates that as long as the distance between ${\bf
x}_n$ and ${\bf x}^*$ is larger than $K\gamma$, the next iteration leads to a decrease at least a
constant step of $d\gamma^2$.

When the parameter $\mu$ is variable, the decrease step of $\|{\bf x}_n-{\bf x}^*\|_2$ is also
variable. The expressions show that $K$ and $d$ increase as the increase of $\mu$. Notice that
$\mu$ must obey
\begin{equation}\label{eq54}
1<\mu\le\frac{2t}{\gamma}\frac{\|{\bf x}_n-{\bf x}^*\|_2}{\displaystyle\max_{{\bf
x}\in\mathbb{R}^N}{\|{\bf P}{\rm sgn}({\bf x})\|_2^2}},
\end{equation}
where the right inequality is necessary to satisfy (\ref{eq99}), which ensures the convergence of
the sequence. During the very beginning of recursions, ${\bf x}_n$ is far from ${\bf x}^*$.
Consequently, $\mu$ satisfying (\ref{eq54}) can be larger, and lead to a faster convergence.
However, as ${\bf x}_n$ gets closer to ${\bf x}^*$ by iterations, $\mu$ satisfying (\ref{eq54}) is
definitely just a little larger than one.

To be emphasized, the actual convergence of iterations can not speed up by
choosing the parameter $\mu$. The value of $\mu$ only impacts the sequence of
\begin{equation}\label{eq55}
\|{\bf x}_{n+1}'-{\bf x}^*\|_2^2=\|{\bf x}_n'-{\bf x}^*\|_2^2-d\gamma^2,
\end{equation}
which is a sequence bounding the actual sequence in the proof of convergence.

\subsection{Convergence Rate}

Theorem 4 tells little about the convergence rate. Considering several inequalities utilized in the
proof, the actual convergence is faster than that of the sequence in (\ref{eq55}). It means that a
lower bound of the convergence rate can be derived in theory.

Corresponding to the variable selection of $\mu$, a sequence $\{{\bf x}_n'\}$ is put forward with
properties
\begin{equation}\label{eq46}
\|{\bf x}_{n+1}'-{\bf x}^*\|_2^2=\|{\bf x}_n'-{\bf x}^*\|_2^2-\gamma^2(\mu_n-1) \max_{{\bf
x}\in\mathbb{R}^N}{\|{\bf P}{\rm sgn}({\bf x})\|_2^2},
\end{equation}
where
\begin{equation}\label{eq45}
\mu_n=\frac{2t\|{\bf x}_n'-{\bf x}^*\|_2}{\gamma\displaystyle \max_{{\bf x}\in\mathbb{R}^N}{\|{\bf
P}{\rm sgn}({\bf x})\|_2^2}}>1.
\end{equation}
Combining (\ref{eq46}) and (\ref{eq45}), the iteration of ${\bf x}_n'$ obeys
\begin{align}\label{eq21}
\|{\bf x}_{n+1}'-{\bf x}^*\|_2^2 =&\|{\bf x}_n'-{\bf x}^*\|_2^2-2\gamma t\|{\bf x}_n'-{\bf x}^*\|_2+\gamma^2\max_{{\bf x}\in\mathbb{R}^N}{\|{\bf P}{\rm sgn}({\bf x})\|_2^2}.
\end{align}
The distance between ${\bf x}_n'$ and ${\bf x}^*$ with variable $\mu$ decreases the most for each
step. Therefore, $\{{\bf x}_n'\}$ has a faster convergence rate compared with sequences satisfying
(\ref{eq55}) with other choices of $\mu$. However, as a theoretical result, it
still converges more slowly than the actual sequence.

Derived from Lemma 2, which gives a rough estimation, Theorem 6 provides a much better lower bound
of the convergence rate.

\begin{lemma}
Supposing $\{{\bf x}_n\}$ is the iterative sequence by $\ell_1$-ZAP, it will take at most
$$
\frac{2(K_{\rm max}-K_{\rm min})}{2t-\frac{1}{K_{\rm min}}\displaystyle{\max_{{\bf x}\in\mathbb{R}^N}{\|{\bf P}{\rm
sgn}({\bf x})\|_2^2}}}\nonumber
$$
steps for $\{{\bf x}_n\}$ to get into the $(K_{\rm min}\gamma)$-neighborhood from the $(K_{\rm
max}\gamma)$-neighborhood of ${\bf x}^*$, where $K_{\rm max}>K_{\rm min}$ and $K_{\rm min}$ must
obey
\begin{equation}\label{eq2}
K_{\rm min}>\frac{1}{2t}\max_{{\bf x}\in\mathbb{R}^N}{\|{\bf P}{\rm sgn}({\bf x})\|_2^2}.
\end{equation}
\end{lemma}

\begin{theorem}
Supposing $\{{\bf x}_n\}$ is the iterative sequence by $\ell_1$-ZAP, it will get into the
$(K_0\gamma)$-neighborhood of ${\bf x}^*$ within at most
$$
\frac{M_0}{t\gamma}+\frac{K_0}{t}
\ln{\left(\frac{M_0}{K_0\gamma}\right)}+\frac{2K_0}{2t-\frac{1}{K_0}\displaystyle{\max_{{\bf
x}\in\mathbb{R}^N}{\|{\bf P}{\rm sgn}({\bf x})\|_2^2}}}
$$
steps. Here $M_0, \gamma$, and $t$ have the same definitions with those in Theorem 4, and $K_0$
must obey
$$
K_0>\frac{1}{2t}\max_{{\bf x}\in\mathbb{R}^N}{\|{\bf P}{\rm sgn}({\bf x})\|_2^2}.
$$
\end{theorem}

The proofs of Lemma 2 and Theorem 6 are postponed to Appendix D and E, respectively.

\subsection{Choice of the Initial Value}

In $\ell_1$-ZAP, the initial value is the least square solution of the under-determined equation,
$$
{\bf x}_0={\bf A}^{\rm T}({\bf AA}^{\rm T})^{-1}{\bf y}.
$$
From Theorem 4 and Corollary 1, one knows that if the initial value obeys ${\bf Ax}_0={\bf y}$, the
iterative sequence $\{{\bf x}_n\}$ is convergent. Therefore, the restriction to the initial value
is to be in the solution space, other than to be the least square solution. However, it is still a
convenient way to initialize using the least square solution.

\section{Convergence in Noisy Scenario}

The convergence of $\ell_1$-ZAP in noisy scenario is analyzed in this section. The main theorem in
noisy scenario is given in Subsection A. In Subsection B, the problem of signal
recovery from inaccurate measurements is discussed. Subsection C shows different choices of initial
value and the impact on the quality of reconstruction. The reconstruction of $p$-compressible
signal by $\ell_1$-ZAP is discussed in Subsection D.

\subsection{Main Result in Noisy Scenario}

Considering the perturbation on measurement vector ${\bf y}$, Theorem 7 is presented to
analyze the convergence of $\ell_1$-ZAP. Similar to Lemma 1, Lemma 3 is proposed at first
corresponding to the noisy case.

\begin{lemma}
Suppose that ${\bf x}\in\mathbb{R}^N$ satisfies $\|{\bf y-Ax}\|_2\le\varepsilon$, with given ${\bf
A}\in\mathbb{R}^{M\times N}$ and ${\bf y}\in\mathbb{R}^M$. ${\bf x}^{\star}$ is the unique solution
of (\ref{l1_eps}). $\|{\bf x}-{\bf x}^{\star}\|_2$ is bounded by a positive number $M_0$. Then
there exists a positive number $t$ depending on ${\bf A}$, ${\bf y}$, $M_0$, and $\varepsilon$,
such that
\begin{equation}\label{eqinlemma4}
\|{\bf x}\|_1-\|{\bf x}^{\star}\|_1\ge t\|{\bf x}-{\bf x}^{\star}\|_2.
\end{equation}
\end{lemma}
\begin{proof}
With the definition of (\ref{eq1}),
(\ref{eqinlemma4}) is equivalent to the following inequality
\begin{align}
&\inf_{\bf x}{\rm g}({\bf x})>0,\quad\text{subject to}\;\|{\bf y-Ax}\|_2\le\varepsilon\;\text{and}\;0<\|{\bf x}-{\bf
x}^{\star}\|_2\le M_0.
\end{align}

Following the proof of Lemma 1, it can be readily proved that Lemma 3 is correct. Notice that here
$$
{\bf u}=\frac{{\bf x}-{\bf x}^{\star}}{\|{\bf x}-{\bf x}^{\star}\|_2}
$$
is not in the null-space of $\bf A$, but a unit vector satisfying $\|{\bf Au}\|_2\le2\varepsilon$.
The remaining procedures are similar. The details of the proof are omitted for short.
\end{proof}

\begin{theorem}
Supposing that ${\bf x}^{\star}$ is the unique solution of (\ref{l1_eps}), sequence $\{{\bf x}_n\}$
satisfies the iterative formula (\ref{eq5}) with conditions
\begin{equation}\label{eq52}
\|{\bf y}-{\bf Ax}_n\|_2\le \varepsilon
\end{equation}
and
\begin{equation}\label{eq53}
\|{\bf x}_n-{\bf x}^{\star}\|_2\le M_0,
\end{equation}
where $M_0$ is a positive constant.
Then the iteration obeys
$$
\|{\bf x}_{n+1}-{\bf x}^{\star}\|_2^2\le \|{\bf x}_n-{\bf x}^{\star}\|_2^2-d \gamma^2,
$$
when
$$
\|{\bf x}_n-{\bf x}^{\star}\|_2\ge K\gamma+C\varepsilon,
$$
where $C=\frac{2}{t}\sqrt{N\lambda}$,
 $K$ and $d$ are defined by (\ref{eq25}) and (\ref{eq26}),
respectively. Here $\mu>1$ is a parameter, $t$ is the positive lower bound in Lemma 3, and
$\lambda$ is the largest eigenvalue of matrix $({\bf AA}^{\rm T})^{-1}$.
\end{theorem}

The proof of Theorem 7 goes in Appendix F.

Theorem 7 indicates that under measurement perturbation with energy less than $\varepsilon$,
the iterative sequence $\{{\bf x}_n\}$ will get into the $(K\gamma+C\varepsilon)$-neighborhood of
${\bf x}^{\star}$. For the fixed original signal and measurement matrix, the precision of ${\bf
x}_n$ approaching ${\bf x}^{\star}$ depends on both the step-size and the noise energy bound. It
means that ${\bf x}_n$ can not get close to the solution ${\bf x}^{\star}$ at any precision by
choosing small step-size, because the noise energy also controls a deviation component,
$C\varepsilon$.

\subsection{Signal Recovery from Inaccurate Measurements}

Corollary 2 indicates the property of signal reconstruction with inaccurate
measurements.
\begin{corollary}
Suppose the original signal is ${\bf x}^{\sharp}\in\mathbb{R}^N$, and the conditions of
(\ref{eq47}) and (\ref{eq49}) are satisfied. There exist real numbers $K>0$, $C'>0$ such that
$\ell_1$-ZAP can be convergent to a $(K\gamma+C'\varepsilon)$-neighborhood of ${\bf x}^{\sharp}$,
i.e., $\ell_1$-ZAP can approach the original signal to some extent under inaccurate measurements.
\end{corollary}

\begin{proof}
Referring to the proof of Corollary 1, it can be readily accepted that the condition (\ref{eq53})
is always satisfied for any index $n$. It is known from Theorem 3 that (\ref{l1_eps}) has a unique
solution under the condition (\ref{eq49}). Consequently, according to Theorem 7, the sequence
$\{{\bf x}_n\}$ finally gets into the neighborhood of ${\bf x}^{\star}$ with the radius $K
\gamma+C\varepsilon$.

Theorem 2 shows that under the condition of (\ref{eq47}), the solution of (\ref{l1_eps}) is not far
from the original signal ${\bf x}^{\sharp}$, with the inequality
\begin{equation}\label{eq61}
\|{\bf x}^{\star}-{\bf x}^{\sharp}\|_2\le C_S\varepsilon.
\end{equation}
Combining Theorem 7, (\ref{eq61}), and the triangle inequality, one sees that the sequence gets
into the neighborhood of ${\bf x}^{\sharp}$ with the radius $K \gamma+(C+C_S)\varepsilon$. Denote
$C'=C+C_S$ and the conclusion of Corollary 2 is drawn.
\end{proof}

\subsection{Initial Values}

Among the assumptions of Theorem 7, a condition of (\ref{eq52}) is assumed to be
satisfied. Considering the recursion (\ref{eq5}), one readily sees that
\begin{equation}\label{eq321}
\|{\bf y}-{\bf Ax}_n\|_2=\|{\bf y}-{\bf Ax}_0\|_2.
\end{equation}
Under the simple condition of
\begin{equation}\label{eq19}
\|{\bf y}-{\bf Ax}_0\|_2\le\varepsilon,
\end{equation}
where ${\bf x}_0$ is not necessarily the least square solution of ${\bf y=Ax}$, it will suffice to
get (\ref{eq52}), which satisfies the condition of Theorem 7.

If the initial value satisfies (\ref{eq19}), by defining ${\bf e}_n={\bf A}({\bf x}_n-{\bf
x}^{\star})$, one has
\begin{align}\label{eq34}
\|{\bf e}_n\|_2&=\|({\bf y}-{\bf Ax}_n)-({\bf y}-{\bf Ax}^{\star})\|_2\nonumber\\
&\le\|{\bf y}-{\bf Ax}_n\|_2+\|{\bf y}-{\bf Ax}^{\star}\|_2 \le2\varepsilon.
\end{align}
Inequality (\ref{eq34}) provides the upper bound of $\|{\bf e}_n\|_2$ and it is used to prove
Theorem 7.

If the iterations begin with the least square solution of the perturbed measurement ${\bf y}$, it
obeys ${\bf y}={\bf Ax}_0$ and according to (\ref{eq321}) one has
$$
\|{\bf y}-{\bf Ax}_n\|_2=0,
$$
which means that (\ref{eq34}) can be modified to
\begin{equation}
\|{\bf e}_n\|_2\le\varepsilon.
\end{equation}
Hence, the parameter $\varepsilon$ can be reduced to a half throughout the proof of Theorem 7.
Therefore, if the initial value is chosen as the least square solution, the neighborhood of
convergence will be smaller, i.e., a better estimation can be reached.

\subsection{Discussions on $p$-compressible Signal}

The original signal is not always absolutely sparse. The reconstruction of compressible signal is
discussed here. Signal ${\bf x}$ is $p$-compressible with magnitude $R$ if the components of ${\bf
x}$ decay as
$$
|x_{(i)}|\le R\cdot i^{-1/p},
$$
where $x_{(i)}$ is the $i$th largest absolute value among the components of
${\bf x}$, and $p$ is a number between $0$ and $1$. Supposing that ${\bf x}_S$ is a best $S$-sparse
approximation to ${\bf x}$, the following inequalities hold \cite{cosamp},
\begin{align}\label{eq101}
\|{\bf x-x}_S\|_1&\le C_p\cdot R\cdot S^{1-1/p},\\
\label{eq102}
\|{\bf x-x}_S\|_2&\le D_p\cdot R\cdot S^{1/2-1/p},
\end{align}
where $C_p=(\frac{1}{p}-1)^{-1}$ and $D_p=(\frac{2}{p}-1)^{-\frac{1}{2}}$.

For a $p$-compressible signal ${\bf x}$, one has
$$
{\bf y}={\bf Ax+e}={\bf Ax}_S+({\bf A}({\bf x-x}_S)+{\bf e}).
$$
By Proposition 3.5 in \cite{cosamp}, the norm of ${\bf A}({\bf x-x}_S)$ can be estimated as
\begin{equation}\label{eq103}
\|{\bf A}({\bf x-x}_S)\|_2\le \sqrt{1+\delta_S}\left(\|{\bf x-x}_S\|_2+\frac{1}{\sqrt{S}}\|{\bf x-x}_S\|_1\right).
\end{equation}
Combining (\ref{eq101}), (\ref{eq102}) and (\ref{eq103}), one has
$$
\|{\bf A}({\bf x-x}_S)+{\bf e}\|_2\le \sqrt{1+\delta_S}(D_p+C_p)\cdot R\cdot S^{1/2-1/p}+\varepsilon.
$$
According to Theorem 7 and Corollary 2, the reconstruction property of $p$-compressible signal by
$\ell_1$-ZAP can be deduced as follows.
\begin{corollary}
Supposing ${\bf x}\in\mathbb{R}^N$ is $p$-compressible signal and the conditions of (\ref{eq47})
and (\ref{eq49}) are satisfied, then the $\ell_1$-ZAP sequence can approach ${\bf x}$ with a
deviation
$$
K\gamma+C'\varepsilon+C'\sqrt{1+\delta_S}(D_p+C_p)\cdot R\cdot S^{1/2-1/p},
$$
where $K$ and $C'$ are the same with those in Corollary 2, and $\varepsilon$ is the
energy bound of observation noise.
\end{corollary}

The non-noisy scenario for compressible signal can be naturally obtained by setting $\varepsilon$
to zero in Corollary 3.

\section{Experiments}

Several experiments are conducted in this section. The performance of $\ell_0$-ZAP and $\ell_1$-ZAP
are shown in Subsection A, compared with several other algorithms for sparse recovery. The
deviations of actual $\ell_1$-ZAP sequence and bound sequences in the proof are illustrated in
Subsection B. In Subsection C, experiment results demonstrate the impacts of the
step-size and the noise level on the signal reconstruction via $\ell_1$-ZAP.

\subsection{Performance of ZAP}

The performances of $\ell_1$-ZAP and $\ell_0$-ZAP are simulated, compared with other sparse
recovery algorithms.

In the experiments, the $M\times N$ matrix ${\bf A}$ is generated with the entries independent and
following a normal distribution with mean zero and variance $1/M$. The support set of original
signal ${\bf x}^*$ is chosen randomly following uniform distribution. The nonzero entries follow a
normal distribution with mean zero. Finally the energy of the original signal is normalized.

For parameters $N=1000$, $S=50$, the probability of exact reconstruction for various number of
measurements is shown as Fig.~\ref{fig1}.
If the reconstruction SNR is higher than a threshold of $40$dB, the trial is regarded as exact reconstruction.
The number of $M$ varies from $140$ to $320$ and each
point in the experiment is repeated 200 times. The step-size of $\ell_1$-ZAP is $5\times10^{-4}$.
The parameters of other algorithms are selected as recommended by respective authors. It
can be seen that for any fixed $M$ from $180$ to $260$, $\ell_0$-ZAP and $\ell_1$-ZAP have higher
probability of reconstruction than other algorithms, which means ZAP algorithms demand fewer
measurements in signal reconstructions. The experiment also indicates that the performance of
$\ell_0$-ZAP is better than $\ell_1$-ZAP, as discussed in Section II.

\begin{figure}
\centering
\includegraphics[width=4in]{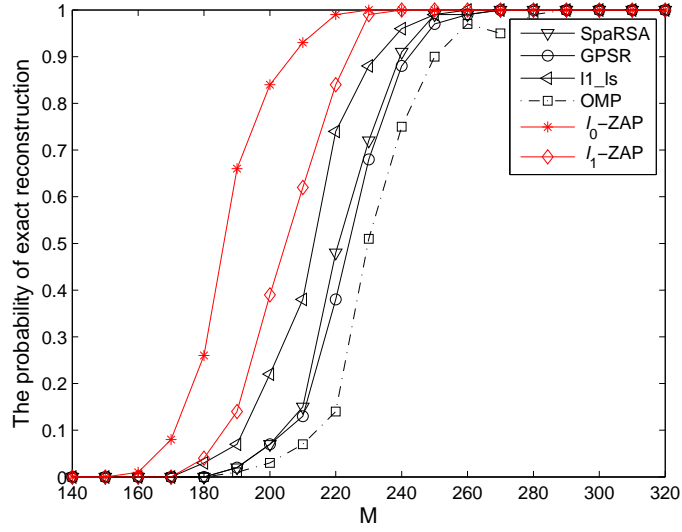}
\caption{Probability of exact reconstruction for various number of measurements, where
$N=1000, S=50$.}\label{fig1}
\end{figure}

For parameters $N=1000$, $M=200$, Fig.~\ref{fig2} illustrates the probability of exact
reconstruction for various sparsity $S$ from $25$ to $70$. All the algorithms are repeated 200
times for each value. The parameters of algorithms are the same as those in the previous
experiment. $\ell_0$-ZAP has the highest probability for fixed sparsity $S$ and $\ell_1$-ZAP is the
second beyond other conventional algorithms. The experiment indicates that ZAP algorithms can
recover less sparse signals compared with other algorithms.

\begin{figure}
\centering
\includegraphics[width=4in]{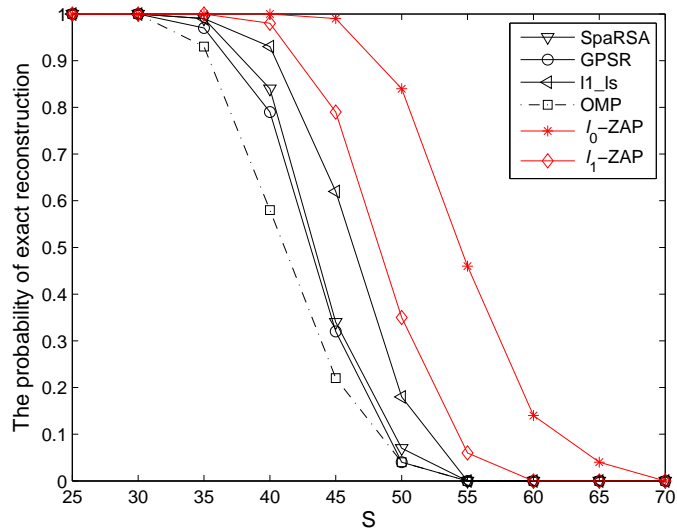}
\caption{Probability of exact reconstruction for various sparsity, where $N=1000, M=200$.}\label{fig2}
\end{figure}

The SNR performance is illustrated in Fig.~\ref{fig3} with the measurement SNR varying from $5$dB to $30$dB and 200
times repeated for each value. The noise is zero-mean white Gaussian and added to the observed vector ${\bf y}$.
The parameters are selected as $N=1000$, $M=200$ and $S=30$. The
parameters of algorithms have the same choice with previous experiments.
The reconstruction SNR and measurement SNR are the signal-to-noise ratios of reconstructed signal $\hat{\bf x}$ and measurement signal ${\bf y}$, respectively.
$\ell_0$-ZAP outperforms other algorithms, while $\ell_1$-ZAP is almost the same as others. The
experiment indicates that $\ell_0$-ZAP has a better performance against noise and $\ell_1$-ZAP does
not have visible defects compared with other algorithms.

\begin{figure}
\centering
\includegraphics[width=4in]{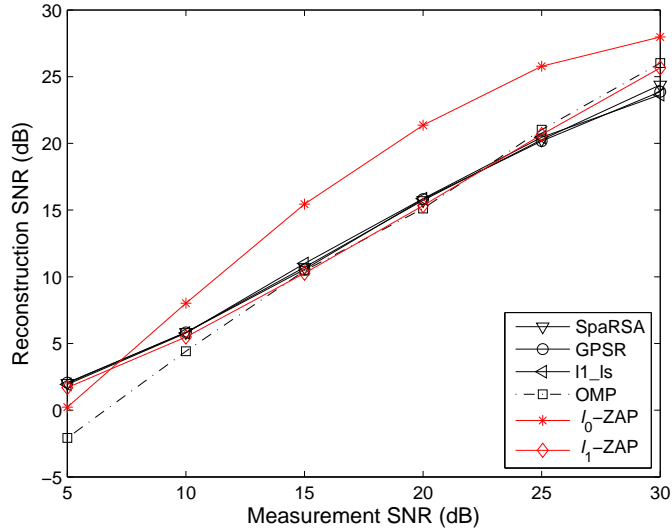}
\caption{Reconstruction SNR versus measurement SNR, where $N=1000, M=200, S=30$.}\label{fig3}
\end{figure}

The experiments above demonstrate that $\ell_1$-ZAP has a better performance compared with
conventional algorithms. $\ell_1$-ZAP demands fewer measurements and can recover signals with
higher sparsity, with similar property against noise. The performance of $\ell_0$-ZAP is better
than $\ell_1$-ZAP.

\subsection{Actual Sequence and Bound Sequences}

According to Theorem 4, the deviation from the actual iterative sequence to the sparse solution is
bounded by the sequence satisfying (\ref{eq55}). In Theorem 4, a sequence with parameter $\mu$ is
utilized to bound the actual sequence and proved to be convergent. As discussed in III-E and F, the
sequence defined in (\ref{eq46}) and (\ref{eq45}) with adaptive $\mu$ approaches the sparse
solution faster than any sequence with constant $\mu$.

The reconstruction SNR curves of the actual sequence and three bound sequences with different choices of
$\mu$ are demonstrated in Fig.~\ref{fig4}.
As can be seen in the figure, the bound sequence with
adaptive $\mu$ is the best estimation among different choices. For a constant $\mu$, the larger one
leads to faster convergence and less precision.

\begin{figure}
\centering
\includegraphics[width=4in]{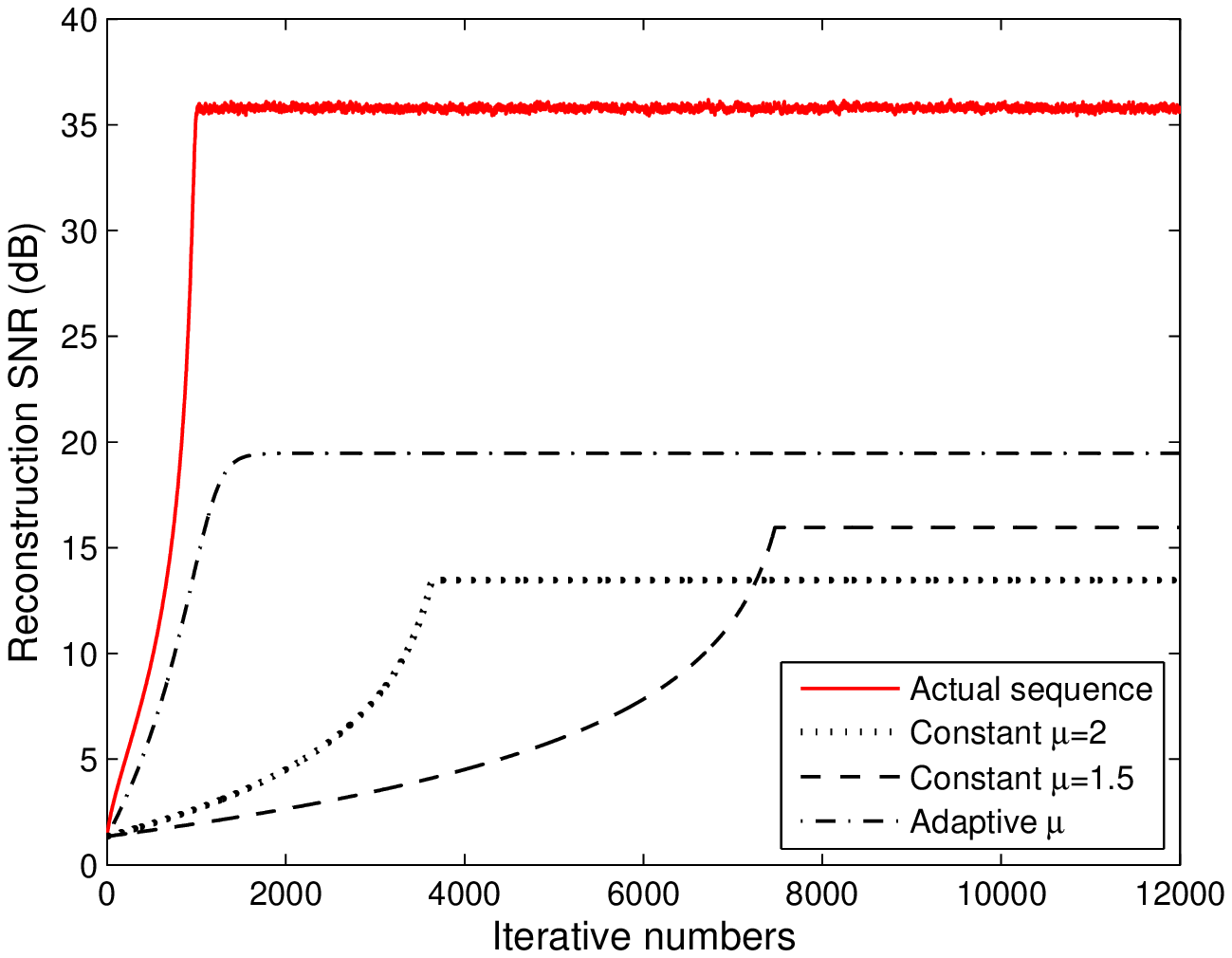}
\caption{Reconstruction SNR of actual sequence and bound sequences with different choices
of $\mu$, where $M=250$, $N=1000$, $S=50$, $\gamma=5\times 10^{-4}$.}\label{fig4}
\end{figure}

\begin{figure}
\centering
\includegraphics[width=4in]{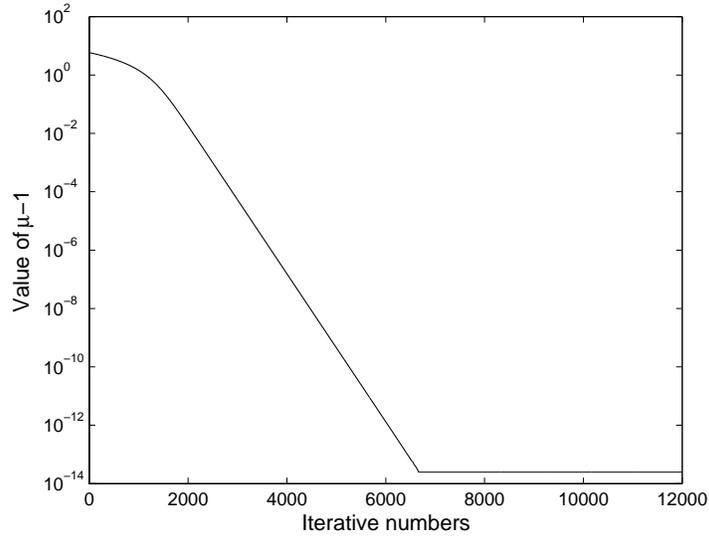}
\caption{The value of $\mu-1$ throughout the iteration for adaptive $\mu$.}\label{fig5}
\end{figure}

For adaptive $\mu$, as illustrated in Fig.~\ref{fig4}, the reconstruction SNR reaches steady-state
after about $2000$ iterations. However, referring to Fig.~\ref{fig5}, the value of $\mu$
keeps decreasing until over $6000$ iterations, though it impacts little to the convergence
behavior. In fact, adaptive $\mu$ will decrease towards $1$ throughout the iteration and never
stop. Nevertheless, the precision of simulation platform limits its variation after it
is below $3\times10^{-14}$.

The deviations of the actual iterative sequence and a bound sequence are both proportional to the
step-size, with the difference in the scale factor. Though the bound is not very
strict, it does well in the proof of the convergence of $\ell_1$-ZAP.

\subsection{About Step-size and Noise}

As proved in Theorem 4, in non-noisy scenario, $\ell_1$-ZAP can reconstruct the original signal at
arbitrary precision by choosing the step-size small enough. Theorem 7 demonstrates that in noisy
scenario the reconstruction SNR is determined by both the step-size and noise level.
Experiment results shown in Fig.~\ref{fig6} verify the analysis.
Each combination of step-size and measurement SNR is simulated 100 times. Experiment results indicate that in non-noisy scenario, the reconstruction SNR increases
 as the decreasing of step-size. In noisy scenario, the reconstruction SNR can not
increase arbitrarily due to the impact of noise. For small step-size, the reconstruction SNR is
mainly determined by noise level. The reconstruction SNR is higher when the measurement SNR is higher. For large step-size, the step-size mainly controls the reconstruction SNR and the reconstruction SNR increase as the decreasing of step-size.

\begin{figure}
\centering
\includegraphics[width=4in]{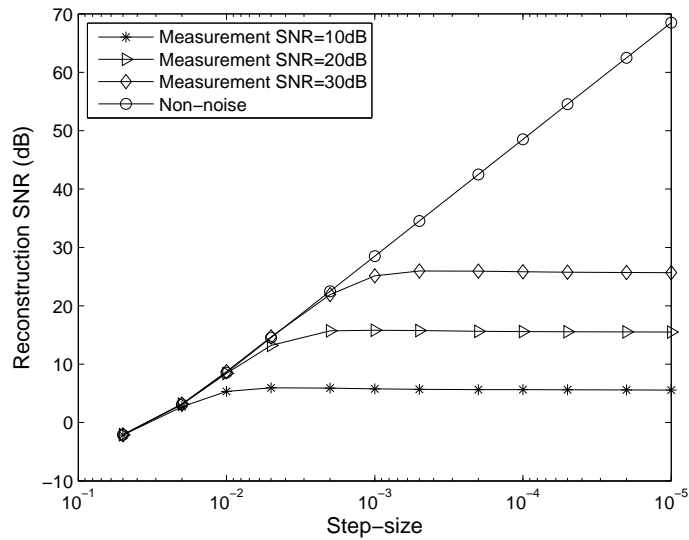}
\caption{Reconstruction SNR versus step-size for various SNR, where $M=150$, $N=1000$, $S=20$.}\label{fig6}
\end{figure}

The figure also offers a way to choose the step-size under noise. It is not necessary to choose the
step-size too small because it benefits little under the impact of noise. For an
estimated reconstruction SNR, the best choice of step-size is the value just entered the flat region.

\section{Conclusion}

This paper provides $\ell_1$-ZAP a comprehensive theoretical analysis. Firstly, the
mentioned algorithm is proved to be convergent to a neighborhood of the sparse solution with the
radius proportional to the step-size of iteration. Therefore, it is non-biased and can approach the
sparse solution to any extent and reconstruct the original signal exactly. Secondly, when the
measurements are inaccurate with noise perturbation, $\ell_1$-ZAP can also approach the sparse
solution and the precision is linearly reduced by the disturbance power. In addition, some related
topics about the initial value and the convergence rate are also discussed. The convergence
property of $p$-compressible signal by $\ell_1$-ZAP is also discussed. Finally, experiments are
conducted to verify the theoretical analysis on the convergence process and illustrate the
impacts of parameters on the reconstruction results.

\appendix
\appendixtitleon
\begin{appendices}
\section{The Proof of Lemma 1}
\label{ProofLemma2}

\begin{proof}
It is to be proved that ${\rm g}({\bf x})$ defined in (\ref{eq1}) has a positive lower bound
respectively for ${\bf x}\in {\mathcal X}_1$ and ${\bf x}\in {\mathcal X}_2$.

For ${\bf x}\in {\mathcal X}_1$, the function ${\rm g}({\bf x})$ is continuous for ${\bf x}$ and
the domain is a bounded closed set. As a basic theorem in calculus, the value of a continuous
function can reach the infinum if the domain is a bounded closed set. As a consequence, there
exists an ${\bf x}_0\in {\mathcal X}_1$, such that ${\rm g}({\bf x}_0)=\inf_{{\mathcal X}_1}{{\rm
g}({\bf x})}$. By the uniqueness of ${\bf x}^*$ and the definition of ${\rm g}({\bf x})$, ${\rm
g}({\bf x})$ is positive in ${\mathcal X}_1$. Then ${\rm g}({\bf x}_0)$ is positive and
this leads to the conclusion that the infimum of ${\rm g}({\bf x})$ is positive in
${\mathcal X}_1$.

On the other hand, it will be proved that ${\rm g}({\bf x})$ has a positive lower bound for ${\bf
x}\in {\mathcal X}_2$.

Any vector in the solution space of ${\bf y=Ax^*}$ can be represented by
\begin{equation}\label{eq3}
{\bf x}={\bf x}^*+r\cdot {\bf u}
\end{equation}
where
\begin{align}
r &= \|{\bf x}-{\bf x}^*\|_2,\\
{\bf u} &=\frac{{\bf x}-{\bf x}^*}{\|{\bf x}-{\bf x}^*\|_2}\label{defineu}
\end{align}
denote the distance and direction, respectively. Considering the definition of $r_0$, one has
\begin{equation}\label{eq7}
({\bf x}^*)^{\rm T}{\rm sgn}({\bf x})=({\bf x}^*)^{\rm T}{\rm sgn}({\bf x}^*)=\|{\bf x}^*\|_1.
\end{equation}
Combining (\ref{eq3}) with (\ref{eq7}), one gets
\begin{equation}
\|{\bf x}\|_1=({\bf x}^*+r\cdot {\bf u})^{\rm T}{\rm sgn}({\bf x})=\|{\bf x}^*\|_1+r\cdot {\bf
u}^{\rm T}{\rm sgn}({\bf x}).
\end{equation}
As a consequence, for $0<r<r_0$, the objective function can be simplified as
\begin{equation}\label{eq35}
{\rm g}({\bf x})=\frac{\|{\bf x}\|_1-\|{\bf x}^*\|_1}{\|{\bf x}-{\bf x}^*\|_2}={\bf u}^{\rm T}{\rm
sgn}({\bf x}) =\sum^N_{k=1}{u_k{\rm sgn}(x_k)}.
\end{equation}

Index set $\mathcal I=\{k~|~x^*_k\neq 0, 1\le k\le N\}$ is the support set of $\bf x^*$. ${\mathcal
I}^{\rm c}$ denotes the complement of ${\mathcal I}$. For $\forall k \in {\mathcal I}$, considering
the definition of $r_0$,
$$
u_k{\rm sgn}(x_k)=u_k{\rm sgn}(x_k^*).
$$
For $\forall k \in {\mathcal I}^{\rm c}$, considering $x_k^*=0$ and the definition of $\bf u$ in
(\ref{defineu}),
$$
u_k{\rm sgn}(x_k)=u_k{\rm sgn}(u_k)=|u_k|.
$$
Consequently, ${\rm g}({\bf x})$ can be rewritten as a function of ${\bf u}$,
\begin{align}
{\rm g}({\bf x})&=\sum_{k\in{\mathcal I}}{u_k{\rm sgn}(x_k)}+\sum_{k\in{\mathcal I}^{\rm c}}{u_k{\rm sgn}(x_k)}\nonumber\\
&={\bf u}_{\mathcal I}^{\rm T}{\rm sgn}({\bf x}^*)+\|{\bf u}_{{\mathcal I}^{\rm c}}\|_1
\triangleq {\rm G}({\bf u}),
\end{align}
where
$$
( \bf u_{\mathcal I})_k=\left\{
    \begin{array}{cl} u_k, & k\in{\mathcal I}; \\
    0, & {\rm elsewhere}, \end{array} \right.
$$
and
$$
{\bf u}_{{\mathcal I}^{\rm c}}={\bf u}-{\bf u}_{\mathcal I}.
$$

It can be seen that ${\rm G}({\bf u})$ is continuous for ${\bf u}$ and the domain of ${\rm G}({\bf
u})$ is $\{{\bf u}\in\mathbb{R}^N~|~\|{\bf u}\|_2=1\}\cap\{{\bf u}\in\mathbb{R}^N~|~{\bf Au}=0\}$.
Since the domain of ${\rm G}({\bf u})$ is the intersection of two closed sets and the first set is
bounded, it is a bounded closed set and ${\rm G}({\bf u})$ can reach the infimum. Then
${\rm g}({\bf x})$ has the minimum. By the uniqueness of ${\bf x^*}$, ${\rm g}({\bf x})$ is
positive, consequently $\inf_{{\mathcal X}_2}{{\rm g}({\bf x})}>0$.

To sum up, the lower bound of ${\rm g}({\bf x})$ is positive for
\begin{align}
{\bf x}&\in {\mathcal X}_1\cup{\mathcal X}_2\nonumber\\
&= \{{\bf x}~|~0<\|{\bf x}-{\bf x}^*\|_2\le
M_0\}\cap\{{\bf x}~|~{\bf y=Ax}\},\nonumber
\end{align}
which completes the proof of Lemma~1.
\end{proof}

\section{The Proof of Theorem 4}
\label{ProofTheorem4}

\begin{proof}
By denoting
$${\bf h}_n = {\bf x}_n-{\bf x}^*$$
as the iterative deviation and subtracting the unique solution ${\bf x}^*$ from both sides
of (\ref{eq5}), one has
\begin{align}\label{eq10}
\|{\bf h}_{n+1}\|_2^2=\|{\bf h}_n-\gamma{\bf P}{\rm sgn}({\bf x}_n)\|_2^2=\|{\bf h}_n\|_2^2-2\gamma{\bf h}_n^{\rm T}{\bf P}{\rm sgn}({\bf x}_n)+\gamma^2\|{\bf P}{\rm sgn}({\bf x}_n)\|_2^2.
\end{align}

According to (\ref{eq57}),
$$
{\bf h}_n^{\rm T}{\bf A}^{\rm T}=({\bf x}_n-{\bf x}^*)^{\rm T}{\bf A}^{\rm T}=0.
$$
Considering
$$
{\bf h}_n^{\rm T}{\bf P}={\bf h}_n^{\rm T}-{\bf h}_n^{\rm T}{\bf A}^{\rm T}({\bf AA}^{\rm
T})^{-1}{\bf A}={\bf h}_n^{\rm T},
$$
$$
({\bf x}^*)^{\rm T}{\rm sgn}({\bf x}_n)\le ({\bf x}^*)^{\rm T}{\rm sgn}({\bf x}^*)=\|{\bf x}^*\|_1
$$
and using Lemma 1, one can shrink the second item of (\ref{eq10}) to
\begin{equation}\label{eq12}
{\bf h}_n^{\rm T}{\rm sgn}({\bf x}_n)\ge \|{\bf x}_n\|_1-\|{\bf x}^*\|_1 \ge t\|{\bf h}_n\|_2.
\end{equation}

Using (\ref{eq12}) and (\ref{eq10}), one has
$$
\|{\bf h}_{n+1}\|_2^2\le\|{\bf h}_n\|_2^2 - 2\gamma t\|{\bf h}_n\|_2 + \gamma^2\|{\bf P}{\rm
sgn}({\bf x}_n)\|_2^2.
$$
Consequently, for any $\mu >1$, if
$$
\|{\bf h}_n\|_2 \ge K\gamma = \gamma\frac{\mu}{2t}\max_{{\bf x}\in\mathbb{R}^N}\|{\bf P}{\rm sgn}({\bf x})\|_2^2,
$$
one has
\begin{align*}
\|{\bf h}_{n+1}\|_2^2 &\le\|{\bf h}_n\|_2^2-d\gamma^2= \|{\bf h}_n\|_2^2-\gamma^2(\mu-1)\max_{{\bf x}\in\mathbb{R}^N}\|{\bf P}{\rm sgn}({\bf x})\|_2^2.
\end{align*}
Theorem 4 is proved.
\end{proof}

\section{The Proof of Theorem 5}

\label{ProofTheorem5}

\begin{proof}
Noticing that $\bf u$ is in the kernel of $\bf A$ and $\bf P$ is a symmetric projection matrix to
the solution space, with (\ref{eq4}) and (\ref{defineu}), one has
$$
\bf Pu = u.
$$
Because $\bf u$ is a unit vector, it can be further derived that
\begin{align}
    {\bf u}^{\rm T}{\rm sgn}({\bf x})
    &=({\bf Pu})^{\rm T}{\rm sgn}({\bf x})=\left<{\bf u}, {\bf P}{\rm sgn}({\bf x})\right>\le\|{\bf P}{\rm sgn}({\bf x})\|_2.\label{usignzoom}
\end{align}
Consider the definition of $t$ in (\ref{eqinlemma2}) and (\ref{eq35}),
\begin{align}
t\le\inf_{{\bf x}\in{{\mathcal X}_1}\cup{{\mathcal X}_2}}{{\rm g}({\bf x})}\le{\bf u}^{\rm
T}{\rm sgn}({\bf x}),\label{eq13}
\end{align}
where ${\mathcal X}_1$ and ${\mathcal X}_2$ are defined in (\ref{eq15}). Combining
(\ref{usignzoom}) and (\ref{eq13}), consequently, the left inequality of (\ref{eq24}) is proved.

Now let's turn to the right inequality of (\ref{eq24}). Because of the property of projection
matrix, ${\bf P}={\bf P}^2$, the eigenvalue of ${\bf P}$ is either $0$ or $1$. For all ${\bf x}$,
one has
\begin{align*}
\|{\bf P}{\rm sgn}({\bf x})\|_2^2&={\rm sgn}^{\rm T}({\bf x}){\bf P}{\rm sgn}({\bf x})\\
&\le \max \{\lambda_{\bf P}\}{\rm sgn}^{\rm T}({\bf x}){\rm sgn}({\bf x}) \\
& = \|{\rm sgn}({\bf x})\|_2^2 \le N,
\end{align*}
where $\left\{\lambda_{\bf P}\right\}$ denotes the eigenvalue set of ${\bf P}$. The arbitrariness
of ${\bf x}$ leads to
$$
\max_{{\bf x}\in\mathbb{R}^N}{\|{\bf P}{\rm sgn}({\bf x})\|_2\le\sqrt{N}}.
$$
Therefore, Theorem 5 is proved.
\end{proof}

\section{The Proof of Lemma 2}

\begin{proof}
For $K_{\rm min}$ satisfying (\ref{eq2}), there exists $\mu'>1$ such that
\begin{equation}
K_{\rm min}=\frac{\mu'}{2t}\max_{{\bf x}\in\mathbb{R}^N}{\|{\bf P}{\rm sgn}({\bf x})\|_2^2}.
\end{equation}
Considering the recursion of sequence $\{{\bf x}_n'\}$ in (\ref{eq21}), it is expected to prove
that
\begin{align}\label{eq8}
\|{\bf x}_n'-{\bf x}^*\|_2^2-2\gamma t\|{\bf x}_n'-{\bf x}^*\|_2+\gamma^2\max_{{\bf
x}\in\mathbb{R}^N}{\|{\bf P}{\rm sgn}({\bf x})\|_2^2}<\left(\|{\bf x}_n'-{\bf x}^*\|_2-\gamma t\left(1-\frac{1}{\mu'}\right)\right)^2,
\end{align}
when
\begin{equation}\label{eq14}
\|{\bf x}_n'-{\bf x}^*\|_2\ge \gamma\cdot \frac{\mu'}{2t}\max_{{\bf x}\in\mathbb{R}^N}{\|{\bf
P}{\rm sgn}({\bf x})\|_2^2}.
\end{equation}
Using (\ref{eq14}), the difference between the left side and the right side of (\ref{eq8}) is
\begin{align}
\gamma^2\left[\max_{{\bf x}\in\mathbb{R}^N}{\|{\bf P}{\rm sgn}({\bf x})\|_2^2}-t^2
\left(1-\frac{1}{\mu'}\right)^2\right]-\frac{2\gamma t}{\mu'}\|{\bf x}_n'-{\bf x}^*\|_2 \le -\gamma^2 t^2\left(1-\frac{1}{\mu'}\right)^2<0.
\end{align}
As a consequence, (\ref{eq8}) holds and it leads to
\begin{equation}\label{eq41}
\|{\bf x}_{n+1}'-{\bf x}^*\|_2<\|{\bf x}_n'-{\bf x}^*\|_2-\gamma t\left(1-\frac{1}{\mu'}\right).
\end{equation}
According to (\ref{eq41}), the quantity of decrease by each step is at least $\gamma t(1-\frac{1}{\mu'})$.

Considering that $\{{\bf x}_n\}$ has a faster convergence rate than that of $\{{\bf x}_n'\}$, and
the trip of $\{{\bf x}_n\}$ is from $(K_{\rm max}\gamma)$-ball to $(K_{\rm
min}\gamma)$-ball, consequently the iteration number is at most
$$
\frac{(K_{\rm max}-K_{\rm min})\gamma}{\gamma t(1-\frac{1}{\mu'})} =\frac{2(K_{\rm max}-K_{\rm
min})}{2t-\frac{1}{K_{\rm min}}\displaystyle{\max_{{\bf x}\in\mathbb{R}^N}{\|{\bf P}{\rm sgn}({\bf x})\|_2^2}}}.
$$
\end{proof}

\section{The Proof of Theorem 6}
\begin{proof}
According to Lemma~2, the iteration number needed from $((n+1)K_0\gamma)$-neighborhood
to $(nK_0\gamma)$-neighborhood is at most
\begin{equation}\label{eqnton1}
\frac{K_0}{t-\frac{1}{2nK_0}\displaystyle{\max_{{\bf x}\in\mathbb{R}^N}{\|{\bf P}{\rm sgn}({\bf x})\|_2^2}}}
=\frac{K_0}{t}\left(1+\frac{1}{\mu_0 n-1}\right),
\end{equation}
where
$$
K_0=\frac{\mu_0}{2t}\max_{{\bf x}\in\mathbb{R}^N}{\|{\bf P}{\rm sgn}({\bf x})\|_2^2}
$$
and $\mu_0$ is larger than $1$.

Assume that $M_0=\|{\bf x}_0-{\bf x}^*\|_2$ obeys
$$
mK_0\gamma<M_0\le(m+1)K_0\gamma,
$$
where $m$ is a positive integer. Utilizing (\ref{eqnton1}), the total iteration number from
$M_0$-neighborhood to $(K_0\gamma)$-neighborhood is at most
\begin{equation}\label{in1}
\frac{K_0}{t}\sum_{n=1}^m\left(1+\frac{1}{\mu_0 n -1}\right),
\end{equation}
which is less than
\begin{equation}\label{in2}
\frac{M_0}{t\gamma}+\frac{K_0}{t}\ln{\left(\frac{M_0}{K_0\gamma}\right)}+\frac{2K_0}{2t-\frac{1}{K_0}\displaystyle{\max_{{\bf
x}\in\mathbb{R}^N}{\|{\bf P}{\rm sgn}({\bf x})\|_2^2}}}.
\end{equation}
Thus Theorem 6 is proved. The relation between (\ref{in1}) and (\ref{in2}) comes from
the following plain algebra,
\begin{align*}
\text{(\ref{in1})}
<&\frac{K_0}{t}\left[m+\frac{1}{\mu_0-1}+\sum_{n=2}^m\frac{1}{\mu_0(n-1)}\right]\\
<&\frac{K_0}{t}\left[m+\frac{1}{\mu_0-1}+\frac{1}{\mu_0}(\ln{(m-1)}+1)\right]\\
=&\frac{K_0}{t}\left[m+\frac{1}{\mu_0}\ln{(m-1)}+\left(\frac{1}{\mu_0-1}+\frac{1}{\mu_0}\right)\right]\\
<&\frac{M_0}{t\gamma}+\frac{K_0}{t}\ln{\left(\frac{M_0}{K_0\gamma}\right)+\frac{K_0}{t}\frac{\mu_0}{\mu_0-1}}
=\text{(\ref{in2})}.
\end{align*}
\end{proof}

\section{The Proof of Theorem 7}
\label{ProofTheorem6}

\begin{proof}
Similar to (\ref{eq10}), by defining ${\bf h}'_n={\bf x}_n-{\bf x}^{\star}$ and ${\bf
e}_n={\bf A}({\bf x}_n-{\bf x}^{\star})$, the deviation iterates by
\begin{align}\label{eq36}
\|{\bf h}'_{n+1}\|_2^2
=\|{\bf h}'_n\|_2^2-2\gamma{{\bf h}'_n}^{\rm T}{\rm sgn}({\bf x}_n)
+2\gamma{\bf e}_n^{\rm T}({\bf AA}^{\rm T})^{-1}{\bf A}{\rm sgn}({\bf x}_n)
+\gamma^2\|{\bf P}{\rm sgn}({\bf x}_n)\|_2^2.
\end{align}
From Lemma 3 and referring to (\ref{eq12}), one has
\begin{equation}\label{eq37}
{{\bf h}'_n}^{\rm T}{\rm sgn}({\bf x}_n)\ge t\|{\bf h}'_n\|_2.
\end{equation}
Next the third item of (\ref{eq36}) will be studied.
By the property of symmetric matrices,
\begin{align}\label{eq39}
& \|{\bf e}_n^{\rm T}({\bf AA}^{\rm T})^{-1}{\bf A}{\rm sgn}({\bf x}_n)\|_2^2\nonumber\\
=& {\rm sgn}^{\rm T}({\bf x}_n){\bf A}^{\rm T}({\bf AA}^{{\rm T}})^{-1}{\bf e}_n{\bf e}_n^{{\rm T}}
({\bf AA}^{{\rm T}})^{-1}{\bf A}{\rm sgn}({\bf x}_n)\nonumber\\
=& {\rm sgn}^{{\rm T}}({\bf x}_n){\bf B}{\rm sgn}({\bf x}_n)\nonumber\\
\le&\max\{\lambda_{\bf B}\}{\rm sgn}^{{\rm T}}({\bf x}_n){\rm sgn}({\bf x}_n)\le N\max\{\lambda_{\bf B}\},
\end{align}
where
$${\bf B}={\bf A}^{{\rm T}}({\bf AA}^{{\rm T}})^{-1}{\bf e}_n{\bf e}_n^{{\rm T}}({\bf AA}^{{\rm T}})^{-1}{\bf A}$$
and $\{\lambda_{\bf B}\}$ denote its eigenvalues. Notice that ${\bf e}_n^{\rm T}({\bf AA}^{{\rm
T}})^{-1}{\bf A}\in \mathbb{R}^{1\times N}$, therefore ${\rm rank}({\bf B})$ is at most
one, and at least $N-1$ of the eigenvalues are zeros. Consequently, one has
\begin{align}\label{eq40}
\max\{\lambda_{\bf B}\}&={\rm tr}({\bf B})
={\rm tr}\left({\bf e}_n^{{\rm T}}({\bf AA}^{{\rm T}})^{-1}{\bf A}{\bf A}^{{\rm T}}({\bf AA}^{{\rm T}})^{-1}
{\bf e}_n\right)\nonumber\\
&={\bf e}_n^{{\rm T}}({\bf AA}^{{\rm T}})^{-1}{\bf e}_n\nonumber\\
&\le\max\{\lambda_{({\bf AA}^{{\rm T}})^{-1}}\}{\bf e}_n^{{\rm T}}{\bf e}_n \le4\varepsilon^2\lambda,
\end{align}
where the last step can be derived by
\begin{align}
\|{\bf e}_n\|_2&=\|({\bf y}-{\bf Ax}_n)-({\bf y}-{\bf Ax}^{\star})\|_2\nonumber\\
&\le\|{\bf y}-{\bf Ax}_n\|_2+\|{\bf y}-{\bf Ax}^{\star}\|_2 \le2\varepsilon.
\end{align}
It can be easily seen that $\max\{\lambda_{({\bf AA}^{{\rm T}})^{-1}}\}$ is positive, if ${\bf
AA}^{{\rm T}}$ is an invertible matrix.

Because ${\bf e}_n^{{\rm T}}({\bf AA}^{{\rm T}})^{-1}{\bf A}{\rm sgn}({\bf x}_n)$ is a scalar,
combining (\ref{eq39}) and (\ref{eq40}), one has
\begin{equation}\label{eq42}
|{\bf e}_n^{{\rm T}}({\bf AA}^{{\rm T}})^{-1}{\bf A}{\rm sgn}({\bf x}_n)|\le
2\varepsilon\sqrt{N\lambda}.
\end{equation}

For $\forall \mu>1$, if
\begin{equation}\label{eq43}
\|{\bf x}_n-{\bf x}^{\star}\|_2\ge \gamma\cdot\frac{\mu}{2t}\max_{{\bf x}\in\mathbb{R}^N}\|{\bf
P}{\rm sgn}({\bf x})\|_2^2 +\varepsilon\cdot\frac{2}{t}\sqrt{N\lambda},
\end{equation}
using (\ref{eq37}), (\ref{eq42}) and (\ref{eq43}), we have
\begin{align}\label{eq44}
& 2({\bf x}_n-{\bf x}^{\star})^{{\rm T}}{\rm sgn}({\bf x}_n)-2{\bf e}_n^{{\rm T}}({\bf AA}^{{\rm
T}})^{-1}{\bf A}{\rm sgn}({\bf x}_n)\ge\gamma\mu\max_{{\bf x}\in\mathbb{R}^N}{\|{\bf P}{\rm sgn}({\bf x})\|_2^2}.
\end{align}

Combining (\ref{eq36}) and (\ref{eq44}), it can be concluded that under the condition of
(\ref{eq43}),
$$
\|{\bf x}_{n+1}-{\bf x}^{\star}\|_2^2\le\|{\bf x}_n-{\bf x}^{\star}\|_2^2-\gamma^2(\mu-1)\max_{{\bf
x}\in\mathbb{R}^N}{\|{\bf P}{\rm sgn}({\bf x})\|_2^2}.
$$
Then Theorem 7 is proved.
\end{proof}
\end{appendices}

\section*{Acknowledgement}

The authors appreciate Jian Jin and three anonymous reviewers for their helpful comments to
improve the quality of this paper. Yuantao Gu wishes to thank Professor Dirk Lorenz for his notification of the projected subgradient method.

\end{document}